\documentclass[reprint,
superscriptaddress,
amsmath,amssymb,
aps,
pra,
floatfix
]{revtex4-2}
\usepackage{pgfplots}
\pgfplotsset{compat=newest}
\usepgfplotslibrary{fillbetween}
\usepackage{tikz}
\usepackage{tikz-3dplot}
\usepackage{quantikz}
\usepackage{xcolor}
\usepackage{graphicx}
\usepackage{grffile} 
\usepackage{braket}
\usepackage{hyperref}
\usepackage{color}
\usepackage{mathrsfs}
\usepackage{booktabs}
\usepackage{multirow}
\usepackage{amsthm}
\usepackage{tabularx}
\newtheorem{theorem}{Theorem}
\newtheorem{proposition}[theorem]{Proposition}
\newtheorem{lemma}[theorem]{Lemma}
\newtheorem{corollary}[theorem]{Corollary}

\begin{document}

\title{Detecting Initial System–Environment Correlations from a Single Observable: Theory, Performance, and Applications}

\author{Ali  Abu-Nada}
\affiliation{Sharjah Maritime Academy, United Arab Emirates}
\author{Russell R. Ceballos}
\affiliation{Department of Physical Sciences, Olive-Harvey College, City Colleges of Chicago, 10001 S Woodlawn Ave, Chicago, IL 60628, USA}
\affiliation{QuSTEAM Initiative, 510 Ronalds St., Iowa City, IA 52245, USA}
\author{Lian-Ao Wu}
\affiliation{Department of Physics, University of the Basque Country UPV/EHU, 48080 Bilbao, Spain}
\affiliation{EHU Quantum Center, University of the Basque Country UPV/EHU, Leioa, Biscay 48940, Spain}
\date{\today}
\begin{abstract}
Initial correlations between a quantum system and its environment can strongly influence open-system dynamics and invalidate standard reduced-dynamics descriptions. Detecting such correlations is therefore an important problem, but most existing approaches require state tomography, multiple system preparations, or some degree of access to the environment.

Here we show that initial system--environment correlations can be certified using only a single calibrated system observable when the interaction is known. For a qubit interacting with an environment via isotropic Heisenberg exchange, we derive exact bounds on the signal $z(t)=\langle\sigma_z^S\rangle(t)$ that hold for all factorized initial states sharing the same calibrated reduced system state. These bounds define a \emph{factorized envelope}: if an observed trajectory exits this envelope, initial correlations are unambiguously certified.

The witness requires only single-axis measurements after a one-time calibration of $\rho_S(0)$ and does not rely on environment access or state tomography. We illustrate the method for several families of correlated initial states and show that the same single-observable logic extends to an exactly solvable pure-dephasing spin--boson model with an infinite environment.

\end{abstract}
\maketitle


\section{Introduction}\label{sec1}

A controllable quantum system is never perfectly isolated. In realistic settings it interacts with surrounding degrees of freedom that are not directly accessible and are therefore treated as an environment~\cite{BreuerPetruccione,AlickiLendi}. This coupling gives rise to open-system dynamics, whose behavior may depend not only on noise generated during the evolution, but also on correlations that are already present between the system and the environment at the initial time~\cite{Pechukas1994,Alicki1995,ShajiSudarshan2005,Laine2010,Modi2012}. Such initial system--environment correlations are of fundamental interest because they can invalidate the standard assumption that reduced dynamics are generated by a completely positive map defined on the full system state space~\cite{Pechukas1994,Alicki1995,ShajiSudarshan2005}. They are also operationally important because they can mimic or mask genuine memory effects: for example, revivals in state distinguishability may originate from correlations already present at $t=0$ rather than from information backflow generated during the subsequent evolution~\cite{Laine2010,Modi2012}.

At the same time, initial system--environment correlations are not merely a nuisance. When deliberately engineered, they can serve as a resource and modify the effective action of environmental noise in useful ways. In particular, appropriately structured initial system--environment entanglement can in some settings suppress decoherence or improve the robustness of quantum information processing (see, e.g., Ref.~\cite{XieWangWu2025}). Whether detrimental or beneficial, however, such correlations must first be identified. This makes the detection of initial system--environment correlations a central problem in the theory and practice of open quantum systems.

If both the system and the environment were experimentally accessible, one could in principle detect correlations by joint measurements or full tomography of the combined state~\cite{Modi2012,Horodecki2009}. In most physical platforms, however, direct access to environmental degrees of freedom is either limited or completely absent. This raises a basic and practically important question:

\begin{quote}
\emph{To what extent can initial system--environment correlations be detected using measurements on the system alone?}
\end{quote}

A number of system-only strategies have been developed. Some methods compare the dynamics generated by multiple initial system preparations and infer correlations by testing whether the effective environment state depends on the chosen preparation~\cite{Modi2012,Ringbauer2015}. Others use trace distance or related distinguishability measures between reduced system states, where an increase in distinguishability can witness either initial correlations or non-Markovian behavior~\cite{Laine2010,Smirne2010,Gessner2011,Dajka2011,Wissmann2013,Li2011,Gessner2014}. These approaches are powerful and conceptually important, but they typically require either many initial preparations or some form of reduced-state tomography. Even when the environment is never measured, the experimentalist often still needs to reconstruct the full Bloch vector of the system by measuring several noncommuting observables.

A complementary line of work assumes prior knowledge of the joint dynamics. When the interaction Hamiltonian is known, one can ask whether an observed reduced evolution is compatible with any factorized initial state of the form $\rho_S(0)\otimes\rho_E(0)$~\cite{Chitambar2015,ByrdChitambar2016}. Recent theoretical results have shown that, under broad conditions, initial correlations are in principle always detectable from system-only data provided the joint dynamics are known and suitable measurements are performed~\cite{Sargolzahi2024}. Experimentally, system-only detection of initial entanglement with an inaccessible environment has also been demonstrated without full two-qubit tomography, by exploiting prior knowledge of the interaction model~\cite{Zhu2022}. These results strongly support the view that initial correlations can be detected without direct environment access.

Despite this progress, most existing protocols still rely on reconstructing the reduced state of the system, or at least on measuring more than one observable. In many experimentally relevant settings---including NMR platforms, trapped ions, superconducting circuits, and NV centers---this tomographic overhead is nontrivial~\cite{Laine2010,Smirne2010,Li2011,Modi2012,Horodecki2009,Amico2008,VandersypenChuang2005}. This motivates a sharper question than the one usually asked:

\begin{quote}
\emph{Can initial system--environment correlations be certified using only a single expectation value of a single system observable, measured along one fixed axis?}
\end{quote}

In this work we answer this question affirmatively for a paradigmatic and exactly solvable model. We consider a qubit system interacting with a qubit environment through the isotropic Heisenberg exchange Hamiltonian and assume that the initial reduced state $\rho_S(0)$ is calibrated once at the beginning of the protocol. After this one-time calibration, no further tomography is performed. Instead, we monitor only the single signal $
z(t)=\langle\sigma_z^S\rangle(t)$,
obtained from repeated measurements of $\sigma_z^S$ at different evolution times. We ask which trajectories of this single observable are compatible with \emph{all} factorized initial states that share the same calibrated reduced state.

Our central result is that, for the isotropic Heisenberg exchange, all such factorized initial states generate a closed and exactly computable time-dependent interval $
\mathcal{R}(t)=[z_{\min}(t),z_{\max}(t)]$,
which we call the \emph{factorized envelope}. This envelope contains every trajectory $z(t)$ obtainable from an uncorrelated preparation consistent with the same $\rho_S(0)$. Consequently, if an experimentally observed trajectory leaves this envelope at any time, then no product initial state can explain the data, and initial system--environment correlations are certified. The witness is one-sided, analytically exact, and experimentally economical: after the initial calibration, it requires only repeated measurements of a single observable along a fixed axis.\\
It is important to emphasize, however, that the witness is inherently one-sided. While any violation of the factorized envelope provides an unambiguous certification of initial system--environment correlations, the converse does not hold. A trajectory that remains entirely inside the envelope is still compatible with correlated initial states and therefore cannot be interpreted as evidence that the initial state was uncorrelated.\\
It is important to emphasize that the proposed witness is inherently
one-sided. If an experimentally observed trajectory leaves the
factorized envelope, then no factorized initial state sharing the same
calibrated reduced system state can reproduce the observation, and
initial system--environment correlations are therefore certified.
Conversely, a trajectory that remains entirely within the envelope does
not establish that the initial state was factorized, since correlated
initial states may still generate trajectories compatible with the
factorized null hypothesis.\\
From the perspective of reduced dynamics, the construction has a natural operational interpretation. When initial correlations may be present, the reduced evolution depends on how the reduced input state is embedded into the joint system--environment space, a dependence commonly described in terms of assignment (embedding) maps~\cite{Pechukas1994,Alicki1995,ShajiSudarshan2005}. In our setting, the standard product assignment
$\mathcal{A}_{\mathrm{prod}}(\rho_S)=\rho_S\otimes\rho_E$
serves as a null hypothesis for uncorrelated preparations. The factorized envelope is then precisely the admissible region generated by this product embedding. A measured violation of the envelope therefore rules out the entire product-assignment class using only a single calibrated observable, without reconstructing $\rho_S(t)$ at later times and without inferring a general assignment map from data.\\
The isotropic Heisenberg exchange model is particularly suitable for this program because it admits both an exact analytic treatment and a transparent geometric interpretation. For factorized initial states, the reduced Bloch vector follows a structured partial-swap trajectory between the system and environment Bloch vectors. This makes it possible to characterize the admissible single-observable dynamics in closed form and to visualize the witness geometrically. We show that the resulting envelope can be violated by several families of correlated initial states, including cases in which the reduced system state is maximally mixed and therefore carries no direct signature of the hidden correlations at $t=0$.\\
Although the isotropic Heisenberg exchange interaction serves as our primary example, the proposed framework is not conceptually restricted to this specific model. The Heisenberg Hamiltonian was chosen because it is a paradigmatic interaction in quantum information science and condensed-matter physics and, importantly, admits an exact analytical characterization of the factorized compatibility envelope. The essential ingredient of our approach is the availability of a sufficiently characterized interaction Hamiltonian that allows the admissible single-observable dynamics generated by all factorized initial states to be determined. Consequently, the same compatibility-based philosophy may be extended to broader classes of system--environment interactions whenever analogous analytical or numerical characterizations are available. Investigating such extensions, including anisotropic exchange models and other exactly solvable open-system Hamiltonians, represents an important direction for future work.
\\
Although our main development uses a minimal two-qubit model, the underlying idea is more general. To demonstrate that the witness is not an artifact of a finite-dimensional environment, we also extend the same single-observable logic to an exactly solvable pure-dephasing spin--boson model with an infinite bath. In that setting, factorized initial states again generate a sharply constrained admissible region for a single system observable, and violations of that region certify initial correlations.\\
The present protocol assumes that the system--environment interaction Hamiltonian has been independently characterized through prior experimental calibration. We emphasize that this assumption defines the intended scope of the method rather than a universal description of all open quantum systems. While it may be restrictive for complex molecular or condensed-matter environments, where microscopic interactions are often only approximately known, it is well justified for many engineered quantum platforms—including superconducting circuits, trapped ions, semiconductor spin qubits, and NV centers—where effective interaction Hamiltonians are routinely identified and calibrated before experiments are performed. Extending the present compatibility-based framework to partially characterized or uncertain interaction models remains an important direction for future research.
\\
The present work differs from previous tomography-free and
reduced-resource approaches in both its null hypothesis and its
measurement requirements. Existing system-only methods commonly infer
initial correlations by comparing several reduced trajectories, preparing
multiple initial system states, monitoring distinguishability measures, or
reconstructing part or all of the reduced density operator. In contrast,
we fix a single calibrated reduced system state and determine the complete
range of one fixed observable that can be generated by every factorized
initial state compatible with that calibration. The resulting
factorized-envelope construction therefore provides an exact
single-observable compatibility test rather than a reconstruction-based
criterion.

The novelty of the present work lies in the combination of several
elements: an analytically exact factorized envelope, optimization over all
unknown environment qubit states, measurements along only one fixed axis
after a one-time calibration, a geometric partial-SWAP interpretation, and
an assignment-map formulation in which the product embedding serves as the
factorized null hypothesis. Reference~\cite{HagenByrd2021}, whose state
families are used in Sec.~\ref{sec:examples}, did not derive the
single-observable factorized envelope or the associated certification
theorem developed here. We additionally examine how the same
compatibility-based reasoning may be implemented in a characterized
infinite-dimensional pure-dephasing environment. More broadly, methods originating from open quantum-system theory and
quantum-inspired dynamical models have also found applications in the
description of complex systems beyond traditional physics, including
network dynamics and social processes~\cite{Alodjants2023}.\\
The main contributions of this work are therefore fourfold.
First, we derive an analytically exact factorized compatibility envelope
for a single system observable under isotropic Heisenberg exchange by
optimizing over all environment qubit states. Second, after a one-time
calibration of the reduced initial system state, the protocol requires
measurements along only one fixed system axis, without reduced-state
tomography during the evolution or direct access to the environment.
Third, we provide both a geometric partial-SWAP interpretation and an
assignment-map formulation of the factorized null hypothesis. Fourth, we
investigate how the same compatibility principle can be applied to a
characterized pure-dephasing spin--boson model with an
infinite-dimensional environment.\\
In contrast to previous tomography-free approaches, the present work derives an exact factorized compatibility envelope for a single calibrated observable, provides a closed-form analytical construction under isotropic Heisenberg exchange, develops a geometric interpretation of the witness, formulates the protocol within the assignment-map framework, and demonstrates its extension to an exactly solvable infinite-dimensional model.\\
The remainder of the paper is organized as follows.
Section~\ref{sec2} describes the operational protocol.
Section~\ref{sec3} introduces the model and notation.
Section~\ref{sec4} derives the factorized envelope.
Section~\ref{sec:examples} presents illustrative examples.
Section~\ref{sec:dephasing} extends the approach to the pure-dephasing spin--boson model.
Finally, we conclude with a discussion.

\section{Procedure}\label{sec2}

In this section we describe the protocol in an operational manner, emphasizing how it is implemented in practice and how the resulting data are interpreted.
We clearly separate what is done \emph{in the laboratory} from what is done 
\emph{in theory} to interpret the experimental data.

\subsection*{What is the basic idea?}

When a quantum system interacts with an environment, correlations can appear in two different ways:
\begin{itemize}
\item correlations can be \emph{created during the evolution}, even if the system and environment were initially uncorrelated;
\item correlations can already \emph{exist at the initial time}, before the evolution begins.
\end{itemize}

In this work we focus on the second case:  
\emph{how can we tell whether correlations were already present at the initial time, using measurements on the system only?}

We assume that the interaction Hamiltonian $H$ is known, and that the initial state of the system $\rho_S(0)$ can be calibrated once. After that, we only measure a single observable of the system, $
z(t) = \langle \sigma_z^S \rangle(t)$.

The key question is therefore:
\begin{quote}
\emph{Given the calibrated system state $\rho_S(0)$ and a known interaction $H$, which values of $z(t)$ are possible if the initial state was uncorrelated?}
\end{quote}

If we can determine this range, then any experimentally observed value outside it proves that the initial state must have been correlated.

For the isotropic Heisenberg exchange, we show that all uncorrelated (product) initial states lead to values of $z(t)$ inside a time-dependent interval,
\begin{equation}
\mathcal{R}(t) = [z_{\min}(t),\, z_{\max}(t)],
\end{equation}
which we call the \emph{factorized envelope}.

\vspace{2mm}
\subsection*{Experimental protocol (Steps 1-3)}

The first three steps describe only what the experimentalist does in the laboratory.  
At this stage, no assumption is made about whether the system and environment are correlated.

\vspace{1mm}
\noindent\textbf{Step 1: Calibrate the system state.}

Prepare the system qubit $S$ in a known reduced state,
\begin{equation}
\rho_S(0) = \tfrac{1}{2}\bigl(\mathbb{I} + \vec{s}\cdot\vec{\sigma}\bigr),
\end{equation}
where $\vec{s}=(s_x,s_y,s_z)$ is the Bloch vector.  
The vector $\vec{s}$ is determined in a separate calibration run, for example by performing tomography at time $t=0$.

This calibration is done only once. Afterward, no further tomography is required.
The environment $E$ starts in some unknown state $\rho_E(0)$, which may or may not be correlated with the system.

\vspace{1mm}
\noindent\textbf{Step 2: Let the system and environment interact.}

In each run of the experiment, the joint system starts in some physical state $\rho_{SE}(0)$ whose system marginal is the calibrated $\rho_S(0)$.  
The joint evolution is governed by the known Hamiltonian $H$,
\begin{equation}
\rho_{SE}(t) = U(t)\,\rho_{SE}(0)\,U^\dagger(t),
\qquad U(t)=e^{-iHt}.
\end{equation}

Operationally, the experimentalist simply prepares the same initial state as reproducibly as possible and allows the system and environment to interact for a time $t$.

\vspace{1mm}
\noindent\textbf{Step 3: Measure a single observable.}

For a chosen evolution time $t$, the experiment is repeated many times:
\begin{enumerate}
\item prepare the same initial joint state;
\item let it evolve for time $t$;
\item measure $\sigma_z^S$ on the system.
\end{enumerate}

Averaging the measurement outcomes gives
\begin{equation}
z(t) = \langle \sigma_z^S \rangle(t)
= \mathrm{Tr}_S[\rho_S(t)\,\sigma_z^S],
\end{equation}
where $\rho_S(t)=\mathrm{Tr}_E[\rho_{SE}(t)]$.
Only this single observable is measured, and the measurement axis is fixed throughout the experiment.

\vspace{2mm}
\subsection*{Theoretical analysis (Steps 4-5)}

We now explain how the experimental data are interpreted.

\vspace{1mm}
\noindent\textbf{Step 4: Define the uncorrelated reference model.}

To test whether the data are compatible with an uncorrelated preparation, we introduce a \emph{theoretical null hypothesis}.  
We assume that the initial joint state was uncorrelated and given by $
\rho_{SE}(0)=\rho_S(0)\otimes\rho_E(0)$,
where $\rho_E(0)$ is arbitrary.

This assumption corresponds to the standard \emph{product assignment map}, which embeds the calibrated system state into a joint state assuming no initial correlations.

Under this assumption, different environment states lead to different trajectories $z(t)$.  
The set of all values obtainable in this way defines the factorized envelope,
\begin{equation}
\mathcal{R}(t)
= \Bigl\{ z(t)\ \Big|\ \rho_{SE}(0)=\rho_S(0)\otimes\rho_E(0) \Bigr\}
= [z_{\min}(t), z_{\max}(t)].
\end{equation}

For the isotropic Heisenberg exchange, the bounds $z_{\min}(t)$ and $z_{\max}(t)$ can be computed analytically, as shown in Sec.~\ref{sec4}.

\vspace{1mm}
\noindent\textbf{Step 5: Compare theory and experiment.}

Finally, the experimentally measured value $\tilde z(t)$ is compared with the theoretical interval $\mathcal{R}(t)$.  
If for some time $t^*$ one finds
\[
\tilde z(t^*)
\notin
\mathcal R(t^*)
\]
then no uncorrelated (product) initial state can explain the observation.  
This certifies the presence of initial system-environment correlations.

\vspace{2mm}

Figure~\ref{fig:witness} provides a schematic illustration of the factorized envelope and how an experimentally observed trajectory can certify initial correlations by exiting this admissible region.

\begin{figure}[t]
    \centering
    \includegraphics[width=1.0\linewidth]{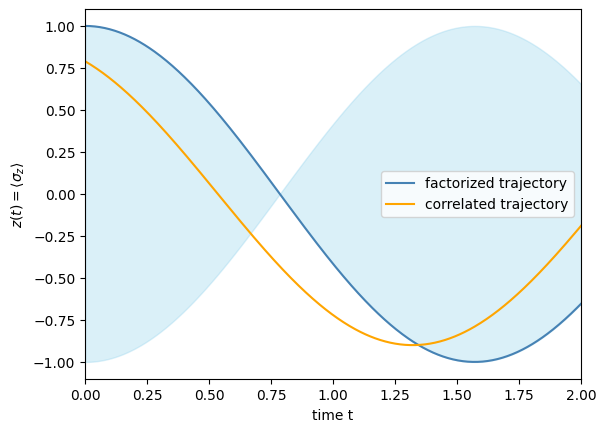}
    \caption{Illustration of the factorized-envelope witness. The shaded region represents the factorized envelope $\mathcal{R}(t)=[z_{\min}(t),z_{\max}(t)]$, containing all trajectories $z(t)=\langle\sigma_z^S\rangle(t)$ obtainable from factorized initial states $\rho_{SE}(0)=\rho_S(0)\otimes\rho_E(0)$. The blue curve shows a representative factorized trajectory lying entirely inside the envelope, while the orange curve corresponds to a correlated initial state that exits the envelope, thereby certifying initial system–environment correlations.}
    \label{fig:witness}
\end{figure}

\section{Model and Definitions}\label{sec3}

We now formalize the setting and notation.

We consider a system qubit $S$ and an environment qubit $E$ with joint Hilbert space
\begin{equation}\label{eq11}
\mathcal{H}_{SE} = \mathcal{H}_S \otimes \mathcal{H}_E,
\qquad \dim\mathcal{H}_S = \dim\mathcal{H}_E = 2.
\end{equation}
The initial joint state $\rho_{SE}(0)$ may be either
\begin{equation}\label{eq12}
\rho_{SE}(0) = \rho_S(0)\otimes\rho_E(0)
\quad\text{(factorized)},
\end{equation}
or
\begin{equation}\label{eq13}
\rho_{SE}(0) \neq \rho_S(0)\otimes\rho_E(0)
\quad\text{(correlated)},
\end{equation}
where the latter may contain classical and/or quantum correlations, including entanglement.

The joint evolution is generated by a known Hamiltonian $H$ on $\mathcal{H}_{SE}$,
with unitary $U(t) = e^{-iHt}$ and
\begin{equation}
\rho_{SE}(t) = U(t)\rho_{SE}(0)U^\dagger(t).
\end{equation}
The reduced state of the system is
\begin{equation}
\rho_S(t) = \mathrm{Tr}_E[\rho_{SE}(t)].
\end{equation}

The calibrated initial state of $S$ is written in Bloch form as
\begin{equation}\label{eq14}
\rho_S(0) = \tfrac{1}{2}\bigl(\mathbb{I} + \vec{s}\cdot\vec{\sigma}\bigr),
\qquad |\vec{s}|\le 1,
\end{equation}
where $\vec{\sigma}=(\sigma_x,\sigma_y,\sigma_z)$ is the Pauli vector and
$\vec{s}\in\mathbb{R}^3$ is the initial Bloch vector. At time $t$ the reduced state is
\begin{equation}\label{eq15}
\rho_S(t) = \tfrac{1}{2}\bigl(\mathbb{I} + \vec{s}(t)\cdot\vec{\sigma}\bigr),
\qquad |\vec{s}(t)|\le 1,
\end{equation}
with Bloch vector $\vec{s}(t)=(s_x(t),s_y(t),s_z(t))$.

The quantity accessible in our protocol is the single expectation value.\\
The measured signal can be written equivalently in the Schr\"odinger and
Heisenberg pictures as
\begin{align}
z(t)
&=
\left\langle \sigma_z^S \right\rangle_t
\nonumber\\
&=
\operatorname{Tr}_S
\left[
\rho_S(t)\sigma_z^S
\right]
\nonumber\\
&=
\operatorname{Tr}_{SE}
\left[
\rho_{SE}(0)
U^\dagger(t)
\left(
\sigma_z^S\otimes\mathbb I_E
\right)
U(t)
\right].
\label{eq:z-signal}
\end{align}
Equivalently, defining the Heisenberg-picture joint observable
\begin{equation}
\sigma_z^S(t)
=
U^\dagger(t)
\left(
\sigma_z^S\otimes\mathbb I_E
\right)
U(t),
\label{eq:heisenberg-observable}
\end{equation}
we have
\begin{equation}
z(t)
=
\operatorname{Tr}_{SE}
\left[
\rho_{SE}(0)\sigma_z^S(t)
\right].
\label{eq:z-heisenberg}
\end{equation}
When the initial state is factorized, we may write the environment as
\begin{equation}\label{eq17}
\rho_E(0) = \tfrac{1}{2}\bigl(\mathbb{I} + \vec{e}\cdot\vec{\sigma}\bigr),
\qquad |\vec{e}|\le 1,
\end{equation}
but in the protocol we never need to know $\vec{e}$ explicitly.

For each fixed $t$ we define the set of all values of $z(t)$ obtainable from factorized initial states with the given calibration:
\begin{equation}\label{eq18}
\mathcal{R}(t)
= \bigl\{\, z(t)\ \big|\ \rho_{SE}(0)=\rho_S(0)\otimes\rho_E(0)\,\bigr\}.
\end{equation}
Our main technical goal is to compute $\mathcal{R}(t)$ explicitly for the isotropic exchange Hamiltonian and to show that it is always an interval
\begin{equation}\label{eq19}
\mathcal{R}(t) = [\,z_{\min}(t), z_{\max}(t)\,].
\end{equation}
Explicit formulas for $z_{\min}(t)$ and $z_{\max}(t)$ are obtained in Sec.~\ref{sec4}.

\subsection{Reduced-dynamics viewpoint: assignment maps and the product admissible set}
\label{sec:assignmentmaps}

A reduced system trajectory is ultimately generated by a joint unitary evolution on $SE$ followed by a partial trace. When initial correlations may be present, a key foundational issue is that the reduced evolution cannot, in general, be represented by a single completely positive (CP) map on the full system state space without additional physical structure; instead, one must specify how a given reduced state $\rho_S(0)$ is embedded into the joint space. This is naturally captured by an \emph{assignment map} (or embedding) $\mathcal{A}$, which assigns to each reduced input state a compatible joint state,
\begin{equation}
\rho_{SE}(0)=\mathcal{A}\!\big(\rho_S(0)\big),
\label{eq:assignment}
\end{equation}
and generates the reduced dynamics
\begin{equation}
\rho_S(t)=\mathrm{Tr}_E\!\left[U(t)\,\mathcal{A}\!\big(\rho_S(0)\big)\,U^\dagger(t)\right].
\label{eq:reduced_from_assignment}
\end{equation}
In general, $\mathcal{A}$ need not be unique, and physical consistency requirements may restrict the \emph{domain} on which a given reduced description is valid. A useful structural condition is \emph{$G$-consistency}: for a specified set of allowed global unitaries $G$, different joint states that share the same reduced state must induce the same reduced evolution under all $U\in G$, ensuring that the reduced dynamics are well defined on the chosen domain (see, e.g., Ref.~\cite{CeballosThesis2017}).

In our protocol, we do \emph{not} attempt to reconstruct $\rho_S(t)$ or infer $\mathcal{A}$ in full generality. Instead, we fix the reduced input state $\rho_S(0)$ once by calibration and monitor only the single expectation value
\begin{equation}
z(t)=\mathrm{Tr}_S\!\big[\rho_S(t)\,\sigma_z^S\big]
=\mathrm{Tr}_{SE}\!\big[\mathcal{A}\!\big(\rho_S(0)\big)\,U^\dagger(t)(\sigma_z^S\!\otimes\!\mathbb{I})U(t)\big].
\label{eq:z_from_assignment}
\end{equation}
This viewpoint makes clear that even without full tomography, one can define a sharp \emph{compatibility set} for factorized preparations by restricting to the standard product assignment
\begin{equation}
\mathcal{A}_{\rm prod}\!\big(\rho_S(0)\big)=\rho_S(0)\otimes\rho_E(0),
\label{eq:product_assignment}
\end{equation}
with arbitrary environment state $\rho_E(0)$. For fixed $\rho_S(0)$ and $t$, the set of all values of $z(t)$ obtainable from \eqref{eq:product_assignment} is exactly the factorized admissible region,
\begin{equation}
\mathcal{R}(t)=\Big\{\,z(t)\ \Big|\ \rho_{SE}(0)=\rho_S(0)\otimes\rho_E(0)\,\Big\}.
\label{eq:R_def_again}
\end{equation}
Because $z(t)$ is affine in the unknown environment parameters under the product assignment, $\mathcal{R}(t)$ is an interval $[z_{\min}(t),z_{\max}(t)]$ whose endpoints are achieved by pure environment states. If an observed trajectory leaves this interval at any time, then no product assignment \eqref{eq:product_assignment} can explain the data, and the initial joint state must have contained correlations. In Sec.~\ref{sec4} we compute $z_{\min/\max}(t)$ analytically for the isotropic Heisenberg exchange.

\section{Heisenberg Exchange Dynamics and Envelope}
\label{sec4}

We now specialize to the isotropic Heisenberg exchange Hamiltonian
\begin{equation}\label{eq20}
H_{\mathrm{ex}}
= J\bigl(\sigma_x^S\sigma_x^E + \sigma_y^S\sigma_y^E + \sigma_z^S\sigma_z^E\bigr)
= J\,\vec{\sigma}_S\cdot\vec{\sigma}_E,
\end{equation}
with factorized initial state
\begin{equation}\label{eq21}
\rho_{SE}(0)=\rho_S\otimes\rho_E,
\quad
\rho_S=\tfrac{1}{2}(\mathbb{I}+\vec{s}\cdot\vec{\sigma}),\quad
\rho_E=\tfrac{1}{2}(\mathbb{I}+\vec{e}\cdot\vec{\sigma}).
\end{equation}

\begin{proposition}
\label{prop:exchange-bloch}
Let the system and environment be two qubits evolving under the isotropic
exchange Hamiltonian
\begin{equation}
H_{\mathrm{ex}}
=
J\,\vec{\sigma}_S\cdot\vec{\sigma}_E .
\end{equation}
Assume that the initial state is factorized,
\begin{equation}
\rho_{SE}(0)=\rho_S\otimes\rho_E,
\end{equation}
where
\begin{equation}
\rho_S=\frac12(\mathbb I+\vec s\cdot\vec\sigma),
\qquad
\rho_E=\frac12(\mathbb I+\vec e\cdot\vec\sigma).
\end{equation}
Then the reduced system state at time \(t\) is
\begin{equation}
\rho_S(t)
=
\frac12
\left[
\mathbb I+\vec s(t)\cdot\vec\sigma
\right],
\end{equation}
with Bloch vector
\begin{equation}
\boxed{
\vec s(t)
=
c^2\vec s
+
d^2\vec e
-
cd\,(\vec s\times\vec e)
}
\label{eq:exchange-bloch-evolution}
\end{equation}
where
\begin{equation}
c=\cos(2Jt),
\qquad
d=\sin(2Jt).
\end{equation}
\end{proposition}

\begin{proof}
Introduce the SWAP operator
\begin{equation}
S
=
\frac12
\left(
\mathbb I+
\vec{\sigma}_S\cdot\vec{\sigma}_E
\right).
\end{equation}
This operator satisfies
\begin{equation}
S^2=\mathbb I,
\end{equation}
and exchanges the two qubits according to
\begin{equation}
S(A\otimes B)S=B\otimes A.
\end{equation}

From the definition of \(S\), we have
\begin{equation}
2S
=
\mathbb I+
\vec{\sigma}_S\cdot\vec{\sigma}_E .
\end{equation}
Therefore
\begin{equation}
\vec{\sigma}_S\cdot\vec{\sigma}_E
=
2S-\mathbb I.
\end{equation}
Hence the Hamiltonian can be written as
\begin{equation}
H_{\mathrm{ex}}
=
J(2S-\mathbb I).
\end{equation}

The unitary evolution operator is then
\begin{align}
U(t)
&=
e^{-iH_{\mathrm{ex}}t}
\nonumber\\
&=
e^{-iJ(2S-\mathbb I)t}
\nonumber\\
&=
e^{iJt}e^{-i2JtS}.
\end{align}
Since \(S^2=\mathbb I\), we obtain
\begin{equation}
e^{-i2JtS}
=
\cos(2Jt)\mathbb I
-
i\sin(2Jt)S.
\end{equation}
The global phase \(e^{iJt}\) cancels in
\(U(t)\rho_{SE}(0)U^\dagger(t)\), so it may be omitted. Thus we write
\begin{equation}
U(t)
=
c\mathbb I-idS,
\end{equation}
where
\begin{equation}
c=\cos(2Jt),
\qquad
d=\sin(2Jt).
\end{equation}

The evolved joint state is
\begin{align}
\rho_{SE}(t)
&=
U(t)(\rho_S\otimes\rho_E)U^\dagger(t)
\nonumber\\
&=
(c\mathbb I-idS)
(\rho_S\otimes\rho_E)
(c\mathbb I+idS)
\nonumber\\
&=
c^2(\rho_S\otimes\rho_E)
+
d^2S(\rho_S\otimes\rho_E)S
-
icd
\left[
S,\rho_S\otimes\rho_E
\right].
\end{align}

Using the SWAP identity,
\begin{equation}
S(\rho_S\otimes\rho_E)S
=
\rho_E\otimes\rho_S.
\end{equation}
Taking the partial trace over the environment gives
\begin{equation}
\rho_S(t)
=
c^2\rho_S
+
d^2\rho_E
-
icd\,
\operatorname{Tr}_E
\left(
[S,\rho_S\otimes\rho_E]
\right).
\label{eq:rhoS-before-comm}
\end{equation}

We now evaluate the commutator term. Since
\begin{equation}
S=
\frac12
\left(
\mathbb I+
\sum_j\sigma_j\otimes\sigma_j
\right),
\end{equation}
only the Pauli part contributes to the commutator. Hence
\begin{align}
\operatorname{Tr}_E
\left(
[S,\rho_S\otimes\rho_E]
\right)
&=
\frac12
\sum_j
\operatorname{Tr}_E
\left(
[\sigma_j\otimes\sigma_j,\rho_S\otimes\rho_E]
\right)
\nonumber\\
&=
\frac12
\sum_j
\operatorname{Tr}_E
\left(
\sigma_j\rho_S\otimes\sigma_j\rho_E
-
\rho_S\sigma_j\otimes\rho_E\sigma_j
\right).
\end{align}
Using
\begin{equation}
\operatorname{Tr}(\sigma_j\rho_E)
=
\operatorname{Tr}(\rho_E\sigma_j)
=
e_j,
\end{equation}
we get
\begin{align}
\operatorname{Tr}_E
\left(
[S,\rho_S\otimes\rho_E]
\right)
&=
\frac12
\sum_j
e_j
(\sigma_j\rho_S-\rho_S\sigma_j)
\nonumber\\
&=
\frac12
\sum_j
e_j[\sigma_j,\rho_S].
\end{align}

Now substitute
\begin{equation}
\rho_S
=
\frac12
\left(
\mathbb I+\sum_k s_k\sigma_k
\right).
\end{equation}
Then
\begin{align}
[\sigma_j,\rho_S]
&=
\frac12
\sum_k s_k[\sigma_j,\sigma_k]
\nonumber\\
&=
\frac12
\sum_k s_k
\left(
2i\sum_\ell\varepsilon_{jk\ell}\sigma_\ell
\right)
\nonumber\\
&=
i
\sum_{k,\ell}
s_k\varepsilon_{jk\ell}\sigma_\ell .
\end{align}
Therefore
\begin{align}
\operatorname{Tr}_E
\left(
[S,\rho_S\otimes\rho_E]
\right)
&=
\frac{i}{2}
\sum_{j,k,\ell}
e_j s_k \varepsilon_{jk\ell}\sigma_\ell
\nonumber\\
&=
-\frac{i}{2}
(\vec s\times\vec e)\cdot\vec\sigma .
\end{align}

Substituting this result into Eq.~\eqref{eq:rhoS-before-comm}, we obtain
\begin{align}
\rho_S(t)
&=
c^2\rho_S
+
d^2\rho_E
-
icd
\left[
-\frac{i}{2}
(\vec s\times\vec e)\cdot\vec\sigma
\right]
\nonumber\\
&=
c^2\rho_S
+
d^2\rho_E
-
\frac{cd}{2}
(\vec s\times\vec e)\cdot\vec\sigma .
\end{align}

Finally, using
\begin{equation}
\rho_S=\frac12(\mathbb I+\vec s\cdot\vec\sigma),
\qquad
\rho_E=\frac12(\mathbb I+\vec e\cdot\vec\sigma),
\end{equation}
we find
\begin{equation}
\rho_S(t)
=
\frac12
\left[
\mathbb I+
\left(
c^2\vec s
+
d^2\vec e
-
cd(\vec s\times\vec e)
\right)
\cdot\vec\sigma
\right].
\end{equation}
Therefore
\begin{equation}
\vec s(t)
=
c^2\vec s
+
d^2\vec e
-
cd(\vec s\times\vec e).
\end{equation}
This completes the proof.
\end{proof}

Taking the $z$-component of Eq.~\eqref{eq:exchange-bloch-evolution} gives
\begin{equation}
z(t)=c^2 s_z + d^2 e_z - cd\bigl(s_x e_y - s_y e_x\bigr),
\label{eq:zHeis}
\end{equation}
where $\vec{s}=(s_x,s_y,s_z)$ and $\vec{e}=(e_x,e_y,e_z)$.

\begin{lemma}
\label{lem:bloch-extrema}
Let $\vec{a}\in\mathbb{R}^3$ and $b\in\mathbb{R}$, and consider
\begin{equation}
f(\vec{e})=\vec{a}\cdot\vec{e}+b,
\qquad
|\vec{e}|\le 1.
\end{equation}
Then the extrema of $f$ over the Bloch ball occur on $|\vec{e}|=1$ and are
\begin{equation}
f_{\max/\min}=b\pm|\vec{a}|.
\end{equation}
\end{lemma}

\begin{proof}
For any $\vec{e}$ with $|\vec{e}|\le 1$,
\begin{equation}
f(\vec{e})-b=\vec{a}\cdot\vec{e}
\le |\vec{a}|\,|\vec{e}|
\le |\vec{a}|,
\end{equation}
so $f(\vec{e})\le b+|\vec{a}|$. The upper bound is achieved for
$\vec{e}=\vec{a}/|\vec{a}|$ when $\vec{a}\neq 0$. A symmetric argument with
$\vec{e}=-\vec{a}/|\vec{a}|$ gives the lower bound $b-|\vec{a}|$.
\end{proof}

\noindent\textbf{Remark.}
Because $f(\vec{e})$ is affine in $\vec{e}$, its extrema over the Bloch ball
occur on the boundary $|\vec{e}|=1$. Thus pure environment states determine the extremal values, while mixed states yield intermediate values inside the interval.

\begin{corollary}
\label{cor:heisenberg-envelope}
For fixed $\vec{s}$ and $t$, all factorized initial states
$\rho_{SE}(0)=\rho_S\otimes\rho_E$ evolving under $H_{\mathrm{ex}}$ satisfy
\begin{equation}
z_{\min}(t)\le z(t)\le z_{\max}(t),
\end{equation}
where
\begin{equation}
z_{\max/\min}(t)
=
c^2s_z
\pm
|d|
\sqrt{
d^2+c^2\bigl(s_x^2+s_y^2\bigr)
}.
\label{eq:z-bounds-heisenberg}
\end{equation}
The absolute value of $d=\sin(2Jt)$ ensures that the interval remains
correctly ordered, $z_{\min}(t)\le z_{\max}(t)$, for all evolution times,
including intervals in which $d<0$.

\end{corollary}

\begin{proof}
From Eq.~\eqref{eq:zHeis} we write
\begin{equation}
z(t)=b(t)+\vec{a}(t)\cdot\vec{e},
\end{equation}
with
\begin{equation}
b(t)=c^2s_z,\qquad
\vec{a}(t)=\bigl(cd\,s_y,\;-cd\,s_x,\;d^2\bigr).
\end{equation}
The norm is
\begin{align}
|\vec{a}(t)|^2
&=
c^2d^2\bigl(s_x^2+s_y^2\bigr)+d^4
\nonumber\\
&=
d^2
\Bigl[
d^2+c^2\bigl(s_x^2+s_y^2\bigr)
\Bigr].
\end{align}
Therefore,
\begin{equation}
|\vec{a}(t)|
=
|d|
\sqrt{
d^2+c^2\bigl(s_x^2+s_y^2\bigr)
}.
\end{equation}
Applying Lemma~\ref{lem:bloch-extrema} with $b=b(t)$ and
$\vec{a}=\vec{a}(t)$ gives
\begin{equation}
\begin{aligned}
z_{\max/\min}(t)
&=
b(t)\pm|\vec{a}(t)|
\\
&=
c^2s_z
\pm
|d|
\sqrt{
d^2+c^2\bigl(s_x^2+s_y^2\bigr)
}.
\end{aligned}
\end{equation}
\end{proof}
\subsection{Geometric interpretation of the exchange dynamics}

\begin{figure*}[t]
\centering
\includegraphics[width=\linewidth]{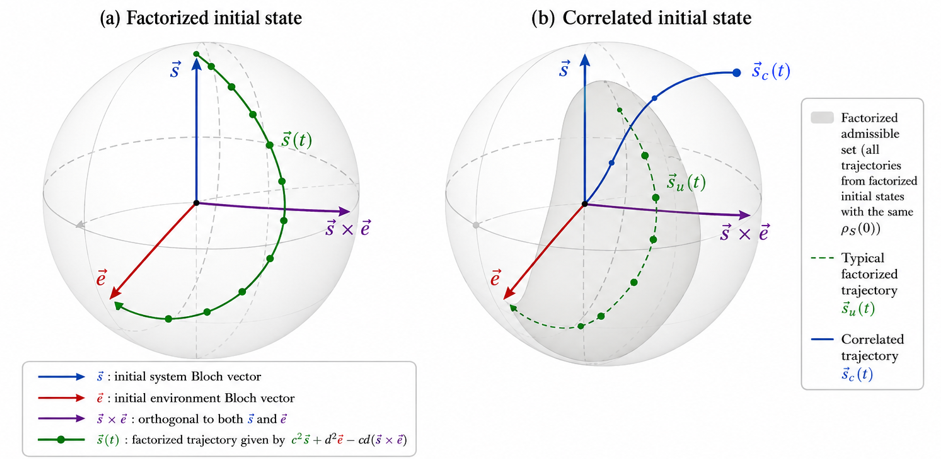}
\caption{(Color online) Geometric interpretation of the reduced dynamics under the isotropic Heisenberg exchange interaction.
(a) For factorized initial states, the reduced Bloch vector evolves according to Eq.~(\ref{eq:exchange-bloch-evolution}), $
\vec{s}(t)=c^{2}\vec{s}+d^{2}\vec{e}-cd(\vec{s}\times\vec{e})$,
where $\vec{s}$ and $\vec{e}$ denote the initial system and environment Bloch vectors, respectively. The resulting trajectory (green) is generally three-dimensional because of the exchange term $\vec{s}\times\vec{e}$.
(b) Schematic illustration of the effect of initial system--environment correlations.
The dashed green curve represents a typical factorized trajectory,
while the shaded gray region schematically represents the family of trajectories compatible with factorized initial states sharing the same calibrated reduced system state.
The solid blue curve illustrates a representative correlated trajectory that leaves this family and may violate the factorized envelope, thereby certifying initial system--environment correlations.}

\label{fig:geometry}
\end{figure*}

Equation~(\ref{eq:exchange-bloch-evolution}) admits a simple geometric
interpretation, illustrated in Fig.~\ref{fig:geometry}. For factorized initial
states, the reduced Bloch vector is given by

\begin{equation}
\vec s(t)
=
c^2\vec s
+
d^2\vec e
-
cd(\vec s\times\vec e),
\end{equation}

where the vectors
$\vec s$,
$\vec e$,
and
$\vec s\times\vec e$
are linearly independent whenever
$\vec s$
and
$\vec e$
are not parallel.

Figure~\ref{fig:geometry}(a) illustrates this evolution.
The blue and red arrows denote the initial system and environment Bloch vectors,
respectively, while the purple arrow represents the vector
$\vec s\times\vec e$, which is orthogonal to both.
The reduced Bloch vector (green curve) evolves from the initial system vector
towards the environment vector while simultaneously acquiring an orthogonal
component generated by the exchange interaction.
Consequently, the trajectory is generally not confined to the plane spanned by
$\vec s$ and $\vec e$.
Instead, it evolves within the three-dimensional subspace generated by
$\{\vec s,\vec e,\vec s\times\vec e\}$.
Only when
$\vec s\times\vec e=\mathbf0$,
that is, when the two Bloch vectors are parallel or antiparallel, does the
trajectory remain planar.

The same dynamics follow naturally from the exchange unitary

\[
U(t)
=
\cos(2Jt)\,\mathbb I
-
i\sin(2Jt)\,S,
\]

where \(S\) is the SWAP operator (up to an overall global phase).
The coefficients
$c=\cos(2Jt)$
and
$d=\sin(2Jt)$
continuously interpolate between the identity operation and a complete SWAP,
causing the reduced Bloch vector to evolve smoothly between the system and
environment Bloch vectors while generating the orthogonal exchange contribution.

Figure~\ref{fig:geometry}(b) illustrates the effect of initial
system--environment correlations.
The dashed green curve represents a typical trajectory generated by a
factorized initial state, whereas the solid blue curve represents a correlated
trajectory.
Initial correlations contribute through the correlation tensor and modify the
reduced dynamics beyond Eq.~(\ref{eq:exchange-bloch-evolution}), producing
trajectories that cannot, in general, be generated by any factorized
preparation having the same calibrated reduced state.
These deviations may lead to violations of the factorized envelope derived in
Corollary~\ref{cor:heisenberg-envelope}, providing the operational basis for
the single-observable witness introduced in this work.
\section{Witness of Initial Correlations}

We can now state the witness formally.

\begin{theorem}[Minimal-observable witness]
\label{thm:minimal-witness}
Let the system evolve under $H_{\mathrm{ex}}$ with known initial
$\rho_S(0)=\tfrac{1}{2}(\mathbb{I}+\vec{s}\cdot\vec{\sigma})$.
For each $t$, let $[z_{\min}(t),z_{\max}(t)]$ be the factorized envelope of Corollary~\ref{cor:heisenberg-envelope}. If for some $t^{*}$ the experimentally observed value $\tilde z(t^{*})$ satisfies
\begin{equation}
\tilde z(t^{*})\notin[z_{\min}(t^{*}),z_{\max}(t^{*})],
\end{equation}
then the initial joint state $\rho_{SE}(0)$ was not factorized. In particular, initial $S$-$E$ correlations are certified.
\end{theorem}

\begin{proof}
Corollary~\ref{cor:heisenberg-envelope} shows that every factorized state
$\rho_{SE}(0)=\rho_S\otimes\rho_E$ with the given $\rho_S$ produces $z(t)$ in the interval $[z_{\min}(t),z_{\max}(t)]$. Thus
\begin{equation}
\mathcal{R}(t)
=\Bigl\{ z(t)\mid \rho_{SE}(0)=\rho_S\otimes\rho_E\Bigr\}
=[z_{\min}(t),z_{\max}(t)].
\end{equation}
If $\tilde z(t^{*})$ lies outside this interval, there is no product state consistent with the observation at $t^{*}$, so $\rho_{SE}(0)$ cannot be factorized.
\end{proof}

\section{Examples: Three Canonical Initial States}
\label{sec:examples}
In the protocol as implemented in the laboratory, only the reduced system state
$\rho_S(0)$ is calibrated experimentally; the joint state $\rho_{SE}(0)$ is
unknown and may or may not be correlated. In the following examples, by
contrast, we \emph{theoretically specify} a particular joint initial state
$\rho_{SE}(0)$ in order to analyze the performance of the witness. The
corresponding calibrated system state is always obtained as
$\rho_S(0)=\mathrm{Tr}_E[\rho_{SE}(0)]$, matching what an experimentalist would
have determined in the initial calibration step. The subsequent comparison
between the correlated trajectory $z^{(\mathrm{corr})}(t)$ and the factorized
envelope $\mathcal{R}(t)$ is therefore carried out using only system data and
the known Hamiltonian, exactly as in the operational protocol.

We now illustrate the witness using the three families of initial states analyzed in Ref.~\cite{HagenByrd2021}. In each case we:
\begin{enumerate}
\item specify the (possibly correlated) initial state $\rho_{SE}(0)$;
\item determine the calibrated reduced state $\rho_S(0)$ and its Bloch vector $\vec{s}(0)$;
\item compute the correlated trajectory $z^{(\mathrm{corr})}(t)$; and
\item compare $z^{(\mathrm{corr})}(t)$ with the factorized envelope $[z_{\min}(t),z_{\max}(t)]$.
\end{enumerate}
Whenever $z^{(\mathrm{corr})}(t)$ leaves the envelope, the witness detects initial correlations using only the single observable $z(t)$.

For convenience we use the dimensionless variable \(Jt\), so that

\[
c=\cos(2Jt),
\qquad
d=\sin(2Jt).
\]
\subsection{Example 1: Maximally entangled Bell-like state}

We first consider a maximally entangled initial state of $S$ and $E$,
\begin{equation}
\rho_{SE}(0)=\ket{\Psi}\!\bra{\Psi},\qquad
\ket{\Psi}
=\tfrac{1}{\sqrt{2}}\bigl(\ket{01}+i\ket{10}\bigr).
\end{equation}
Tracing out the environment gives a maximally mixed reduced state
\begin{equation}
\rho_S(0)=\mathrm{Tr}_E[\rho_{SE}(0)]
=\tfrac{1}{2}\,\mathbb{I},\qquad
\vec{s}(0)=\vec{0},\qquad
z(0)=0.
\end{equation}
Thus, based solely on system information at $t=0$, this preparation is
indistinguishable from an uncorrelated product state with the same $\rho_S(0)$.

\paragraph*{Reduced system trajectory.}
We use the convention
\begin{equation}
\sigma_z\ket{0}=\ket{0},
\qquad
\sigma_z\ket{1}=-\ket{1}.
\end{equation}
For the initial Bell-like state
\begin{equation}
\ket{\Psi}
=
\frac{1}{\sqrt{2}}
\left(
\ket{01}+i\ket{10}
\right)
\end{equation}
and the partial-SWAP unitary
\begin{equation}
U(t)=c\,\mathbb I-id\,S,
\qquad
c=\cos(2Jt),
\qquad
d=\sin(2Jt),
\end{equation}
the evolved joint state is
\begin{align}
\ket{\Psi(t)}
&=
U(t)\ket{\Psi}
\nonumber\\
&=
\frac{1}{\sqrt{2}}
\left[
(c+d)\ket{01}
+i(c-d)\ket{10}
\right].
\label{eq:bell-evolved}
\end{align}
Tracing out the environment therefore gives
\begin{equation}
\rho_S(t)
=
\begin{pmatrix}
p_0(t) & 0\\[4pt]
0 & 1-p_0(t)
\end{pmatrix},
\qquad
p_0(t)
=
\frac{1}{2}
\left[
1+\sin(4Jt)
\right].
\label{eq:bell-reduced}
\end{equation}
Hence the correlated single-axis observable is
\begin{align}
z^{(\mathrm{corr})}(t)
&=
\left\langle
\sigma_z^S
\right\rangle(t)
\nonumber\\
&=
p_0(t)-\left[1-p_0(t)\right]
\nonumber\\
&=
\sin(4Jt).
\label{eq:z-bell}
\end{align}
\paragraph*{Factorized comparison envelope.}
To apply the witness, we compare $z^{(\mathrm{corr})}(t)$ with the set of all
\emph{factorized} preparations that are consistent with the same calibrated
$\rho_S(0)=\mathbb{I}/2$. For any such state, the initial system Bloch vector
satisfies $\vec{s}(0)=\vec{0}$, so Eq.~\eqref{eq:zHeis} reduces to
\begin{equation}
z(t)=d^2\,e_z,\qquad d=\sin(2Jt),\qquad e_z\in[-1,1].
\end{equation}
Optimizing over $e_z$ yields the factorized envelope
\begin{equation}
z_{\min/\max}(t)=\pm\sin^2(2Jt),
\end{equation}
and therefore every product initial state must satisfy
\[
-\sin^2(2Jt)\ \le\ z(t)\ \le\ \sin^2(2Jt).
\]

\paragraph*{Witness violation and figure interpretation.}
Figure~\ref{fig:bell-envelope} compares the correlated trajectory
$z^{(\mathrm{corr})}(t)=\sin(4Jt)$ (black curve) with the factorized
envelope $[z_{\min}(t),z_{\max}(t)]$ (blue shaded region). At the witness
time
\begin{equation}
t^*=\frac{\pi}{8J},
\end{equation}
we have
\begin{equation}
\sin(4Jt^*)=1,
\qquad
\sin^2(2Jt^*)=\frac{1}{2}.
\end{equation}
Therefore,
\begin{equation}
z^{(\mathrm{corr})}(t^*)=1,
\end{equation}
whereas every factorized initial state compatible with the same calibrated
reduced state satisfies
\begin{equation}
z(t^*)\in
\left[
-\frac{1}{2},
\frac{1}{2}
\right].
\end{equation}
The red marker in Fig.~\ref{fig:bell-envelope} therefore lies above the
upper factorized bound. Since no product initial state with the same
calibrated $\rho_S(0)$ can reproduce this value, the observation certifies
initial system--environment correlations using only the single observable
$z(t)$.

\paragraph*{Experimental note.}
In practice, $z(t)=\langle\sigma_z^S\rangle(t)$ is obtained by repeatedly
preparing the system and performing a projective measurement of $\sigma_z^S$ at
time $t$, then averaging the binary outcomes. This procedure uses only a single
measurement axis and is standard in NMR, trapped ions, superconducting qubits,
and NV centers.

\begin{figure}[t]
    \centering
    \includegraphics[width=1.0\linewidth]{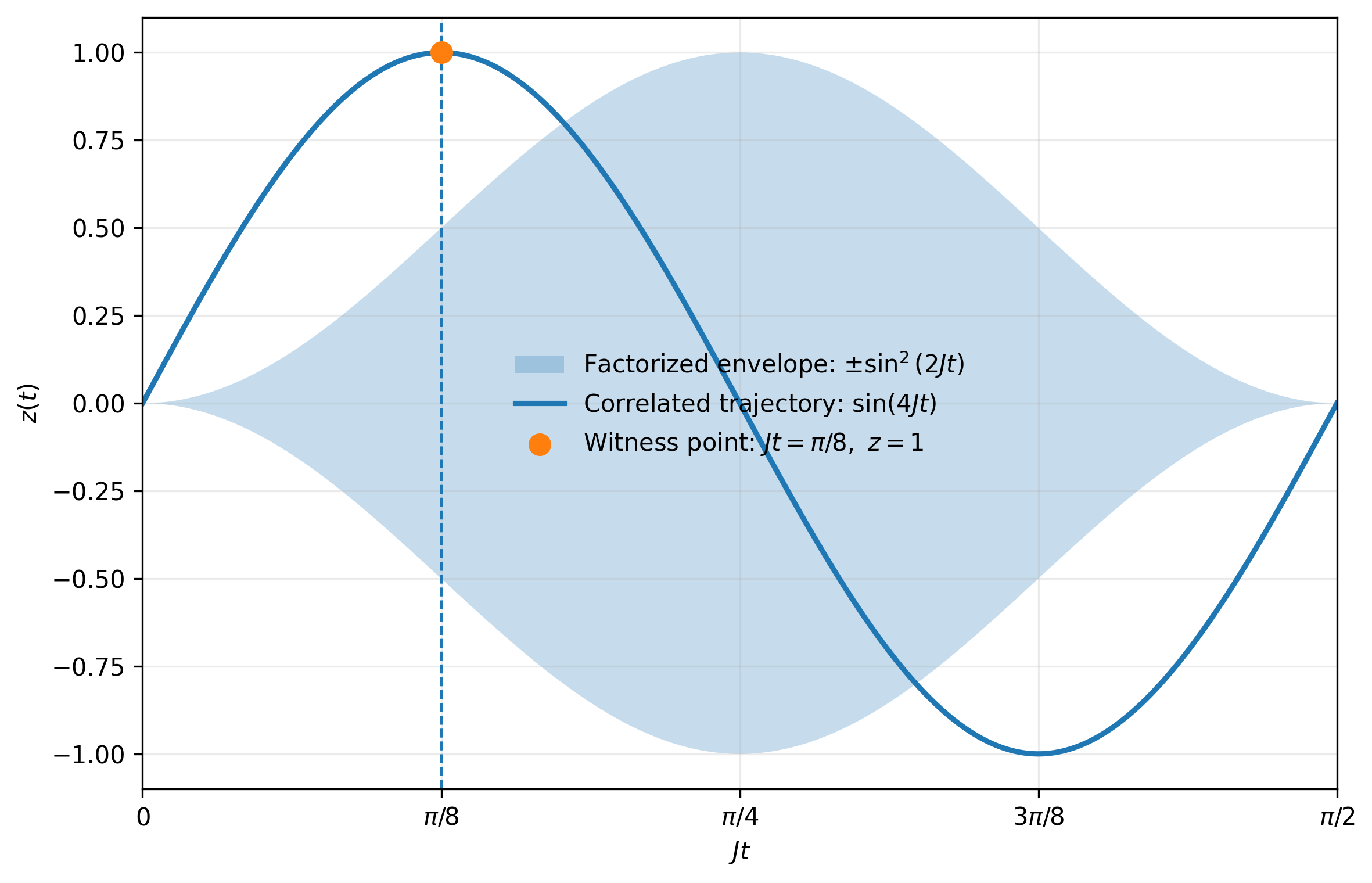}
 
    \caption{(Color online) Example~1.
    The black curve shows the corrected correlated trajectory
    $z^{(\mathrm{corr})}(t)=\sin(4Jt)$ for the maximally entangled
    initial state. The blue shaded region is the factorized envelope
    $z_{\min/\max}(t)=\pm\sin^2(2Jt)$ for all product states with
    $\rho_S(0)=\mathbb{I}/2$. At $t^*=\pi/(8J)$, indicated by the red
    marker, the correlated value is $z^{(\mathrm{corr})}(t^*)=1$,
    whereas the factorized interval is $[-1/2,1/2]$. The trajectory
    therefore lies above the upper factorized bound, certifying initial
    correlations using only $z(t)$.}
    \label{fig:bell-envelope}
\end{figure}
\subsection{Example 2: Mixture of a product state and an entangled state}

We next consider a family of mixed initial states interpolating between a product state and the maximally entangled state of Example~1:
\begin{equation}
\rho_{SE}(0)
= p\,\ket{01}\!\bra{01}
+ (1-p)\ket{\Psi}\!\bra{\Psi}, \qquad 0\le p\le 1,
\end{equation}
with $\ket{\Psi}$ as defined above. Tracing out the environment gives
\begin{equation}
\rho_S(0)
= \begin{pmatrix}
\tfrac{1+p}{2} & 0\\[4pt]
0 & \tfrac{1-p}{2}
\end{pmatrix},
\qquad \vec{s}(0)=(0,0,p),
\qquad z(0)=p,
\end{equation}
so the calibrated system state contains a finite polarization $p$ along $z$.

\paragraph*{Correlated trajectory.}
Using the singlet-triplet decomposition as before, evolving under
$H_{\mathrm{ex}}=J\,\vec{\sigma}_S\cdot\vec{\sigma}_E$ and tracing out $E$ yields the correlated single-axis observable
\begin{equation}
z^{(\mathrm{corr})}(t)
=
p\cos(4Jt)
+
(1-p)\sin(4Jt).
\label{eq:z-example2}
\end{equation}
which continuously connects to Example~1 when $p=0$.

\paragraph*{Factorized comparison envelope.}
To construct the witness, we compare $z^{(\mathrm{corr})}(t)$ with all
\emph{factorized} initial states consistent with the same calibrated
$\rho_S(0)$. For any such product state, $\vec{s}(0)=(0,0,p)$, so
Eq.~\eqref{eq:zHeis} reduces to
\begin{equation}
z(t)=c^2p+d^2e_z,\qquad e_z\in[-1,1],
\end{equation}
with $c=\cos(2Jt)$ and $d=\sin(2Jt)$. Optimizing over $e_z$ yields the envelope
\begin{equation}
z_{\min/\max}(t)=\cos^2(2Jt)\,p\pm\sin^2(2Jt),
\label{eq:envelope-ex2}
\end{equation}
and therefore every factorized initial state with the same $\rho_S(0)$ must satisfy
\[
\cos^2(2Jt)\,p - \sin^2(2Jt)\ \le\ z(t)\ \le\
\cos^2(2Jt)\,p + \sin^2(2Jt).
\]

\paragraph*{Witness violation.}
A convenient witness time is
\[
t^*=\frac{\pi}{8J},
\]
for which
\[
z^{(\mathrm{corr})}(t^*)=1-p,
\qquad
z_{\min/\max}(t^*)
=
\left[
\frac{p-1}{2},
\frac{p+1}{2}
\right].
\]
Violation occurs whenever
\[
1-p>\frac{p+1}{2},
\]
which is equivalent to
\[
p<\frac13.
\]
Thus, for
\[
0\le p<\frac13,
\]
the correlated trajectory lies above the upper factorized bound at
$t^*$, thereby certifying initial system--environment correlations using
only the single observable $z(t)$. For larger values of $p$, no
violation occurs at this particular witness time, although it may occur
at other times.

\paragraph*{Figure interpretation.}
Figure~\ref{fig:mixed-envelope} shows the correlated trajectories
$z^{(\mathrm{corr})}(t)$ (black curves) together with the corresponding
factorized envelopes (blue shaded regions) for three representative values of
$p$. The panels illustrate how decreasing $p$ increases the weight of the
entangled component and eventually produces a clear envelope crossing. For
$p=0.6$ the correlated curve remains inside the envelope at $t^*$; for $p=0.2$
and $p=0.1$ the correlated value at the red marker lies strictly outside the envelope, certifying initial correlations with a single observable.

\begin{figure*}[t]
    \centering
    \includegraphics[width=1.0\linewidth]{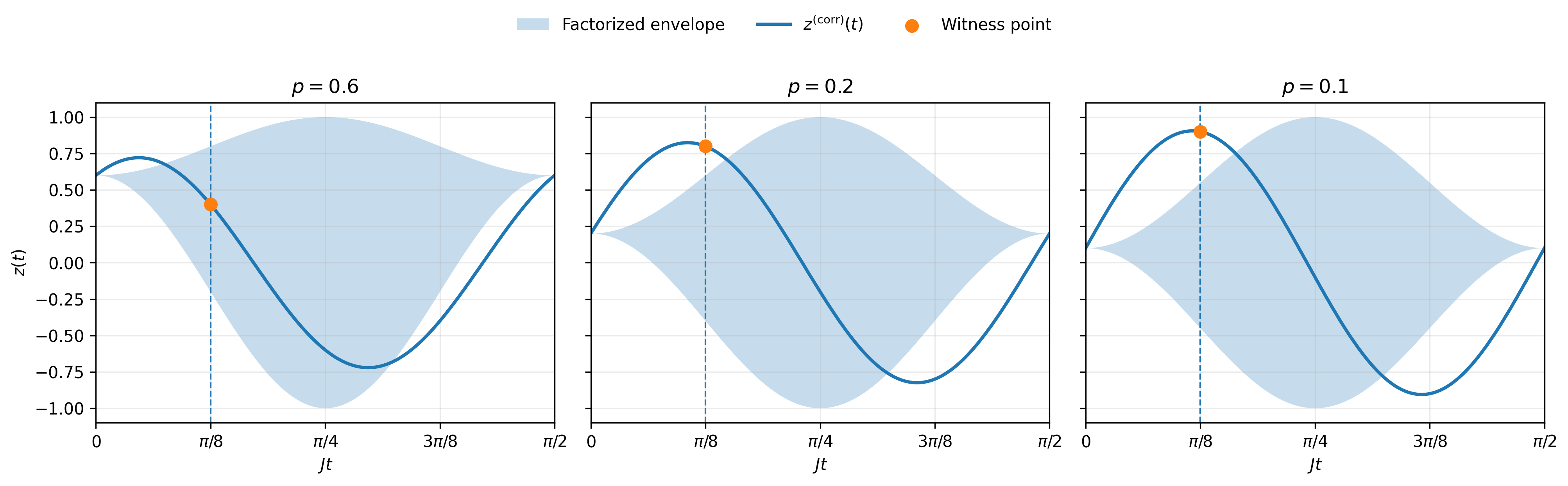}

    \caption{(Color online) Example~2.
    The blue shaded region shows the factorized envelope
    $z_{\min/\max}(t)=p\cos^{2}(2Jt)\pm\sin^{2}(2Jt)$
    for all product initial states compatible with the calibrated reduced
    state $\rho_S(0)$. The black curves show the corrected correlated
    trajectories
    $z^{(\mathrm{corr})}(t)=
    p\cos(4Jt)+(1-p)\sin(4Jt)$
    for different values of the mixing parameter $p$.
    At the witness time
    $t^{*}=\pi/(8J)$ (red marker),
    the correlated value is
    $z^{(\mathrm{corr})}(t^{*})=1-p$,
    while the corresponding factorized upper bound is
    $z_{\max}(t^{*})=(1+p)/2$.
    For $p<1/3$, the correlated trajectory exceeds the upper factorized
    bound, thereby certifying initial system--environment correlations
    using only the single observable $z(t)$.}
  
    \label{fig:mixed-envelope}
\end{figure*}
\subsection{Example 3: Mixture with the maximally mixed state}

Finally, we consider a different type of mixture in which the entangled state of
Example~1 is diluted by the maximally mixed two-qubit state:
\begin{equation}
\rho_{SE}(0)
= p\,\frac{\mathbb{I}_4}{4}
+ (1-p)\ket{\Psi}\!\bra{\Psi},
\qquad 0\le p\le 1,
\end{equation}
where $\ket{\Psi}$ is the maximally entangled state used previously. Tracing out
the environment yields
\begin{equation}
\rho_S(0)
= p\,\frac{\mathbb{I}_2}{2} + (1-p)\frac{\mathbb{I}_2}{2}
=\frac{\mathbb{I}_2}{2},
\end{equation}
so the calibrated system state is maximally mixed, with $\vec{s}(0)=\vec{0}$ and
$z(0)=0$, regardless of $p$.

\paragraph*{Correlated trajectory.}
The maximally mixed part is invariant under all unitary dynamics and therefore
does not contribute to $z(t)$. Only the entangled component evolves nontrivially,
and using the result of Example~1 we obtain

\begin{equation}
z^{(\mathrm{corr})}(t)
=
(1-p)\sin(4Jt).
\label{eq:z-example3}
\end{equation}

which continuously interpolates between the fully entangled trajectory
($p=0$) and a flat signal ($p=1$).

\paragraph*{Factorized comparison envelope.}
Since $\rho_S(0)=\mathbb{I}_2/2$ for all $p$, the factorized initial states have
$\vec{s}(0)=\vec{0}$ and the envelope reduces to the same form as in
Example~1:
\begin{equation}
z_{\min/\max}(t)=\pm\sin^2(2Jt).
\end{equation}
Therefore, every factorized initial state compatible with the same calibration
must satisfy
\[
-\sin^2(2Jt)\ \le\ z(t)\ \le\ \sin^2(2Jt).
\]

\paragraph*{Witness violation and figure interpretation.}
A violation occurs whenever
\begin{equation}
\bigl|z^{(\mathrm{corr})}(t)\bigr|
=
(1-p)|\sin(4Jt)|
>
\sin^2(2Jt).
\label{eq:example3-violation}
\end{equation}
Using
\begin{equation}
|\sin(4Jt)|
=
2|\sin(2Jt)\cos(2Jt)|,
\end{equation}
and assuming $\sin(2Jt)\neq0$, the violation condition can be written as
\begin{equation}
2(1-p)|\cos(2Jt)|
>
|\sin(2Jt)|,
\end{equation}
or equivalently,
\begin{equation}
|\tan(2Jt)|
<
2(1-p).
\label{eq:example3-condition}
\end{equation}
For every $p<1$, the right-hand side of
Eq.~\eqref{eq:example3-condition} is strictly positive. Therefore, the
condition is satisfied for sufficiently small but nonzero values of $t$.

Equivalently, near $t=0$,
\begin{equation}
z^{(\mathrm{corr})}(t)
=
(1-p)\sin(4Jt)
=
4J(1-p)t+\mathcal{O}(t^3),
\end{equation}
whereas the factorized envelope behaves as
\begin{equation}
\sin^2(2Jt)
=
4J^2t^2+\mathcal{O}(t^4).
\end{equation}
Thus, for every $p<1$, the correlated signal grows linearly in $t$ near
the origin, while the width of the factorized envelope grows
quadratically. Consequently, there always exists a sufficiently small
nonzero time at which the correlated trajectory lies outside the
factorized envelope. Although the reduced state is maximally mixed and
contains no evidence of correlations at $t=0$, later measurements of the
single observable $z(t)$ can therefore reveal the hidden initial
system--environment correlations for every $p<1$.\\
Figure~\ref{fig:werner-envelope} illustrates this behavior. The blue shaded
regions indicate the factorized envelope, while the black curves show
$z^{(\mathrm{corr})}(t)$ for three representative values of $p$. For every
$p<1$, the correlated trajectory leaves the factorized envelope at
sufficiently small but nonzero times, because the correlated signal grows
linearly near $t=0$, whereas the factorized envelope grows quadratically.
However, at the particular witness time
$t^*=\pi/(8J)$ indicated by the red marker, a violation occurs only when
$p<1/2$. Consequently, the panels with $p=0.4$ and $p=0.1$ show a violation
at the marked witness time, while the $p=0.8$ trajectory can still leave the
envelope at earlier, sufficiently small nonzero times. As $p\rightarrow1$,
the correlated signal vanishes and the initial state approaches the maximally
mixed product state, consistently eliminating the witness violation.

\begin{figure*}[t]
    \centering
    \includegraphics[width=1.0\linewidth]{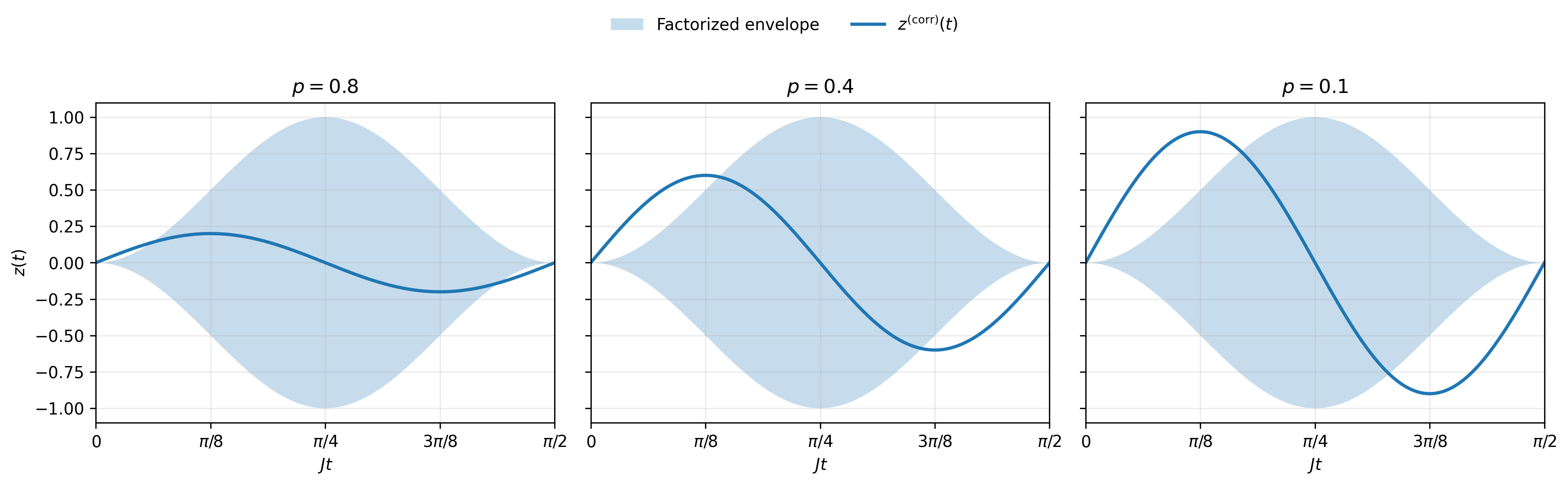}
   
    \caption{(Color online) Example~3.
    The blue shaded region shows the factorized envelope
    $z_{\min/\max}(t)=\pm\sin^{2}(2Jt)$
    for all product initial states compatible with the calibrated reduced
    state $\rho_S(0)=\mathbb{I}/2$. The black curves show the corrected
    correlated trajectories
    $z^{(\mathrm{corr})}(t)=(1-p)\sin(4Jt)$
    for different values of the mixing parameter $p$.
    At the witness time
    $t^{*}=\pi/(8J)$,
    the correlated value is
    $z^{(\mathrm{corr})}(t^{*})=1-p$,
    while every factorized initial state satisfies
    $z(t^{*})\in[-1/2,\,1/2]$.
    For the chosen witness time
$t^*=\pi/(8J)$,
the correlated trajectory exceeds the upper factorized bound whenever
$p<1/2$, thereby certifying initial system--environment correlations
using only the single observable $z(t)$.}

    \label{fig:werner-envelope}
\end{figure*}
\section{Extension to an Exactly Solvable Infinite-Bath Model}
\label{sec:dephasing}

While the isotropic Heisenberg exchange provides a minimal and fully controllable
setting for developing the single-observable witness, it is natural to ask
whether the same logic survives in more conventional open-system models with
infinite environments. In this section, we demonstrate that the witness extends
naturally to the exactly solvable pure-dephasing spin--boson model, which is a
standard benchmark in the theory of open quantum systems.

\subsection{Pure-dephasing spin--boson Hamiltonian}

We consider the Hamiltonian [Eq.~(4.28) of Ref.~\cite{BreuerPetruccione}]
\begin{equation}
H
= \frac{\omega_0}{2}\sigma_z
+ \sum_k \omega_k b_k^\dagger b_k
+ \sigma_z \sum_k \left( g_k b_k^\dagger + g_k^* b_k \right),
\label{eq:dephasing-H}
\end{equation}
where $\sigma_z$ acts on the system qubit and $\{b_k,b_k^\dagger\}$ are bosonic
operators describing an infinite environment. This model generates pure
dephasing: the system populations remain constant, while the coherences decay in
time. Importantly, the reduced dynamics admit a closed-form solution for
factorized initial states, making this model analytically tractable.

\subsection{Reduced dynamics for factorized initial states}

Assuming an initially factorized state
$\rho_{SE}(0)=\rho_S(0)\otimes\rho_E(0)$, the reduced system state at time $t$ is
given by
\begin{equation}
\rho_S(t)=
\begin{pmatrix}
\rho_{00}(0) & \rho_{01}(0)\,e^{-\Gamma(t)} \\
\rho_{10}(0)\,e^{-\Gamma(t)} & \rho_{11}(0)
\end{pmatrix},
\label{eq:dephasing-solution}
\end{equation}

where $\Gamma(t)\ge0$ is the decoherence function corresponding to the
assumed bath model (for example, a thermal bath with known spectral
density and temperature). Throughout this section, the null hypothesis is
that the initial state is factorized,
\[
\rho_{SE}(0)=\rho_S(0)\otimes\rho_E(0),
\]
the reduced system state $\rho_S(0)$ is fixed by the initial calibration,
and the bath is completely characterized through the known decoherence
function $\Gamma(t)$. Its explicit form is not required for our purposes;
we use only that $\Gamma(0)=0$ and
$e^{-\Gamma(t)}\in[0,1]$ for all $t$.

Since $\langle\sigma_z\rangle(t)$ is constant under pure dephasing, we focus on a
single transverse observable,
\begin{equation}
x(t)=\langle\sigma_x\rangle(t)
=\mathrm{Tr}_S[\rho_S(t)\sigma_x].
\end{equation}
From Eq.~\eqref{eq:dephasing-solution}, one obtains
\begin{equation}
x(t)=x(0)\,e^{-\Gamma(t)},
\label{eq:x-dephasing}
\end{equation}

where
$x(0)=\langle\sigma_x\rangle(0)$
is fixed by the initial calibration of the reduced system state.
Consequently, under the factorized null hypothesis, the complete admissible
prediction is determined jointly by the calibrated initial coherence
$x(0)$ and the known decoherence function $\Gamma(t)$.

\subsection{Factorized envelope and witness}

For every factorized initial state compatible with the calibrated reduced
system state,
\begin{equation}
|x(t)|
\le
|x(0)|\,e^{-\Gamma(t)}.
\label{eq:spinboson-envelope}
\end{equation}

This inequality defines the admissible factorized prediction for the
calibrated initial coherence. Therefore, any experimentally observed value
satisfying
\[
|x(t)|>|x(0)|e^{-\Gamma(t)}
\]
cannot arise from any factorized initial state sharing the same calibrated
reduced system state and therefore certifies initial
system--environment correlations.\\
This inequality defines a time-dependent \emph{factorized envelope} for the
single observable $x(t)$ under the pure-dephasing dynamics. It represents the
entire set of values compatible with the standard product assignment
$\rho_{SE}(0)=\rho_S(0)\otimes\rho_E(0)$ for an infinite bosonic bath.

Unlike the exchange model, the pure-dephasing dynamics admit no further
optimization over environmental degrees of freedom beyond the decoherence
function $\Gamma(t)$. Consequently, the factorized envelope takes a particularly
simple form, while any deviation from it directly reflects the presence of
initial system--environment correlations.

If, for some time $t^{*}$, an experimentally observed value satisfies
\begin{equation}
|\tilde{x}(t^{*})|
>
|x(0)|e^{-\Gamma(t^{*})},
\end{equation}
then no factorized initial state sharing the same calibrated reduced
system state and evolving under the assumed bath model can reproduce the
observation. Such a violation therefore certifies initial
system--environment correlations.\\
This example demonstrates that the single-observable witness is not restricted
to finite environments or exchange-type interactions. Even in a paradigmatic
infinite-bath model with analytically solvable reduced dynamics, factorized
initial states generate a sharply constrained admissible region for a single
system observable. Initial correlations modify the coherence dynamics and can
drive the signal outside this region, providing a clear and experimentally
accessible signature of correlated preparations.

\section{Performance Analysis of the Minimal-Observable Witness}
\label{sec:performance}

The factorized-envelope witness developed in the previous sections provides
a simple criterion for certifying initial system--environment correlations
using measurements of only a single calibrated observable.
In this section we discuss its operational performance, practical
requirements, and inherent limitations.
Rather than introducing a new detection criterion,
our goal is to clarify how the witness should be interpreted in realistic
experimental settings and to identify the situations in which it provides
its greatest advantage.

\subsection{One-Sided Nature of the Witness}

The factorized-envelope construction is intentionally designed as a
\emph{one-sided witness}.
For every calibrated reduced state $\rho_S(0)$ and known interaction
Hamiltonian, the interval
$\mathcal{R}(t)=[z_{\min}(t),z_{\max}(t)]$
contains all trajectories that can arise from factorized initial states.
Consequently,

\begin{equation}
z(t)\notin\mathcal{R}(t)
\quad\Longrightarrow\quad
\rho_{SE}(0)\neq
\rho_S(0)\otimes\rho_E(0),
\end{equation}

providing an unambiguous certification of initial
system--environment correlations.

The converse, however, does not generally hold.
A correlated initial state may still generate a trajectory lying entirely
inside the factorized envelope,
because the correlation contribution to the reduced dynamics can be
indistinguishable from that produced by an appropriate factorized
environment state.
Remaining inside the envelope therefore cannot be interpreted as evidence
that the initial state was uncorrelated.

This behavior is characteristic of witness-based methods throughout
quantum information theory.
For example, entanglement witnesses certify entanglement whenever their
inequality is violated, while many entangled states remain undetected.
Similarly, Bell inequalities provide sufficient but not necessary
conditions for nonlocality.
The present protocol should therefore be interpreted in exactly the same
spirit: violation certifies initial correlations, whereas absence of
violation simply leaves the possibility of hidden correlations open.

The advantage of this approach is that every detected event is rigorous.
Whenever the observed trajectory exits the factorized envelope, the
existence of initial correlations follows directly from the incompatibility
between the measured dynamics and every admissible factorized preparation.
No optimization, reconstruction of the reduced state, or access to the
environment is required.

\subsection{Experimental Resource Requirements}

The principal motivation for the present protocol is the reduction of
experimental overhead.
Most existing system-only detection methods require repeated state
preparation together with reconstruction of the reduced density matrix at
multiple evolution times.
Such procedures generally involve measurements of several noncommuting
observables and the subsequent implementation of state tomography.

By contrast, the present protocol requires only a one-time calibration of
the initial reduced state $\rho_S(0)$.
Once this calibration has been completed, all subsequent measurements are
performed along a single fixed axis,
namely

\begin{equation}
z(t)=
\langle\sigma_z^S\rangle(t).
\end{equation}

No reconstruction of $\rho_S(t)$ is performed during the evolution,
and no measurements on the environment are required.

The simplification comes at the cost of assuming that the interaction
Hamiltonian is known.
Knowledge of the joint dynamics makes it possible to compute the
factorized envelope analytically, which then serves as the reference model
against which experimental observations are compared.
Consequently, the protocol is particularly attractive in physical
platforms where the interaction Hamiltonian has already been accurately
characterized, including NMR systems, trapped ions, superconducting
circuits, and other highly controlled quantum devices.

The present witness therefore complements rather than replaces existing
system-only detection methods.
Protocols based on trace-distance growth or full state tomography require
fewer assumptions regarding the interaction model but demand
substantially greater measurement resources.
Conversely, our protocol exchanges detailed Hamiltonian knowledge for a
significant reduction in experimental complexity.
\begin{table}[t]
\caption{Comparison of representative approaches for detecting initial
system--environment correlations.}
\label{tab:comparison}
\centering
\scriptsize
\renewcommand{\arraystretch}{1.15}

\begin{tabular}{lccc}
\hline
Method & Ham. & Tomo. & Repeat \\
\hline
Trace distance & No & Yes & Yes\\
Tomography & No & Yes & Yes\\
This work & Yes & No$^{a}$ & No\\
\hline
\end{tabular}

\vspace{1mm}

$^{a}$After one-time calibration of $\rho_S(0)$.
\end{table}
Table~\ref{tab:comparison} highlights the complementary nature of the
present protocol.
Unlike tomography-based approaches, our method does not require repeated
reconstruction of the reduced density operator during the evolution.
Instead, after a one-time calibration of the initial reduced state, all
subsequent measurements are performed on a single fixed observable.
This simplification is achieved by assuming that the system--environment
Hamiltonian is known, allowing the factorized compatibility envelope to be
computed in advance.
The proposed witness should therefore be viewed as complementary to
existing methods rather than as a universal replacement.
It is particularly advantageous in well-characterized quantum platforms,
where accurate Hamiltonian models are available and reducing experimental
measurement overhead is a primary objective.

\subsection{Sensitivity to Measurement Uncertainty}

In practice, every measured expectation value is affected by finite
sampling statistics and experimental imperfections.
Suppose that the experimentally observed signal is

\begin{equation}
z_{\mathrm{exp}}(t)
=
z(t)+\delta z(t),
\end{equation}

where $\delta z(t)$ represents the combined contribution of statistical
fluctuations and measurement noise.

The witness remains applicable provided the experimental uncertainty is
properly incorporated into the comparison with the theoretical envelope.
Instead of comparing the measured value directly with
$\mathcal{R}(t)$,
one compares the corresponding confidence interval

\begin{equation}
z_{\mathrm{exp}}(t)
\pm
\Delta z(t)
\end{equation}

with the factorized region.
Only when the confidence interval lies entirely outside
$\mathcal{R}(t)$
can initial system--environment correlations be certified with the chosen
confidence level.

The robustness of the witness therefore depends primarily on the
separation between the observed trajectory and the factorized envelope.
Large violations remain detectable even in the presence of moderate
measurement uncertainty, whereas trajectories lying very close to the
boundary require improved statistical precision.
This behavior is common to most experimentally implemented witness
protocols and does not affect the theoretical validity of the
certification criterion.

\subsection{Scope and Limitations}

The present work establishes the compatibility-envelope construction for
two exactly solvable open-system models.
The isotropic Heisenberg exchange provides an analytically tractable
finite-dimensional benchmark in which the factorized envelope can be
derived in closed form.
The exactly solvable spin--boson model considered in the following section
demonstrates that the same single-observable logic extends naturally to an
infinite-dimensional environment.

We emphasize, however, that these models are intended to establish the
general methodology rather than to represent every physically relevant
environment.
In particular, the spin--boson example describes pure-dephasing dynamics,
whereas more general dissipative environments involving energy exchange
remain to be investigated.\\
Another important limitation follows directly from the one-sided nature of the proposed witness. Specifically, the criterion is sufficient, but not necessary, for detecting initial system--environment correlations. Whenever the observed trajectory exits the factorized compatibility envelope, the initial state cannot have been factorized, thereby providing a rigorous certification of initial correlations. Conversely, trajectories that remain within the factorized envelope do not imply that the initial state was factorized, since correlated initial states may also generate dynamics consistent with the compatibility bounds. The present protocol should therefore be interpreted as a one-sided certification tool rather than a test for the absence of initial system--environment correlations. Developing stronger witnesses capable of reducing the number of correlated states that remain undetected, while preserving the simplicity of single-observable measurements, constitutes an important direction for future work. A quantitative characterization of the detection efficiency of the present witness, namely the fraction of correlated initial states that produce observable violations of the factorized envelope, also lies beyond the scope of the present work.\\
Although every violation certifies initial system--environment
correlations, correlated states whose trajectories remain inside the
factorized envelope cannot be identified by the present protocol alone.
Developing stronger witnesses capable of reducing the number of correlated
states that remain undetected by the witness, while preserving the
simplicity of single-observable measurements, constitutes an interesting
direction for future work.
A quantitative characterization of the detection efficiency of the present
witness, namely the fraction of correlated initial states that produce
observable violations of the factorized envelope, lies beyond the scope of
the present work.
Such an analysis would require sampling correlated initial
system--environment states according to an appropriate statistical measure
and evaluating the corresponding detection probability.
This represents an interesting direction for future research and would
provide additional insight into the practical performance of the proposed
single-observable witness.
Despite these limitations, the proposed protocol demonstrates that
initial system--environment correlations can be rigorously certified
without state tomography, without repeated measurement of noncommuting
observables, and without direct access to the environment.
This substantial reduction in experimental resources makes the method an
attractive complement to existing system-only approaches for detecting
initial correlations in open quantum systems.

\section{Conclusion}

We have presented a single-observable protocol for certifying initial
system--environment correlations in open quantum systems. For a qubit
interacting with an environment through the isotropic Heisenberg exchange,
we showed that a one-time calibration of the reduced system state, together
with measurements of a single system observable, is sufficient to detect
initial correlations without requiring state tomography, multiple state
preparations, or direct access to the environment.

The central result is the derivation of an exact factorized envelope that
bounds the dynamics of all product initial states compatible with the
calibrated reduced system state. Any experimentally observed trajectory
leaving this envelope excludes every factorized preparation sharing the same
reduced state and therefore certifies the presence of initial
system--environment correlations. A geometric interpretation of the reduced
dynamics clarifies the origin of this constraint, while the examples
demonstrate that the witness successfully detects correlations for several
representative families of initial states.

The proposed protocol is inherently a one-sided certification method.
Violation of the factorized envelope provides a rigorous and experimentally
accessible witness of initial correlations, whereas remaining inside the
envelope does not exclude their presence. The protocol therefore provides a
sufficient, but not necessary, criterion for certifying correlated initial
states.

Finally, we extended the same operational idea to the exactly solvable
pure-dephasing spin--boson model with an infinite environment, where an
analogous coherence envelope can be derived. This shows that the underlying
principle is not restricted to finite-dimensional environments or
exchange-type interactions, but reflects a more general structural property
of reduced dynamics generated from factorized initial states.

Our results demonstrate that surprisingly limited measurement resources can
provide rigorous information about initial system--environment correlations.
They also suggest a broader class of single-observable certification
protocols for other interactions, higher-dimensional environments, and more
general open quantum systems, bringing the experimental detection of initial
correlations closer to practical implementation.
\appendix
\section{Deviation from Partial Swap for Correlated Initial States}\label{app:deviation}

For factorized initial states $\rho_{SE}(0)=\rho_S\otimes\rho_E$, the reduced
system Bloch vector obeys Eq.~\eqref{eq:exchange-bloch-evolution},
\begin{equation}
\vec{s}(t)=c^2\vec{s}+d^2\vec{e}-cd\,(\vec{s}\times\vec{e}),
\label{eq:pswap}
\end{equation}
which depends only on $(\vec{s},\vec{e})$ and the Heisenberg parameters $(c,d)$.\\
Equation~\eqref{eq:pswap} therefore provides a complete description of the
reduced dynamics whenever the initial state is factorized. In particular,
all trajectories belonging to the factorized envelope derived in the main
text are generated solely by varying the unknown environment Bloch vector
$\vec{e}$ while keeping the calibrated system Bloch vector $\vec{s}$ fixed.
To understand how initial system--environment correlations modify this
picture, we now consider the most general two-qubit initial state.\\
For correlated initial states, it is convenient to use the full Bloch decomposition
\begin{equation}
\rho_{SE}(0)
=\tfrac{1}{4}\Bigl(\mathbb{I}\otimes\mathbb{I}
+ \vec{s}\cdot\vec{\sigma}\otimes\mathbb{I}
+ \mathbb{I}\otimes\vec{e}\cdot\vec{\sigma}
+ \sum_{j,k} C_{jk}\,\sigma_j\otimes\sigma_k \Bigr),
\label{eq:chi}
\end{equation}
where $C_{jk}=\mathrm{Tr}[\rho_{SE}(0)\,\sigma_j\otimes\sigma_k]$ is the correlation tensor.

For product states one has
\begin{equation}
C_{jk}=s_j e_k,
\end{equation}
so we write
\begin{equation}
C_{jk}=s_j e_k+T_{jk},
\end{equation}

where
\begin{equation}
T_{jk}=C_{jk}-s_j e_k
\end{equation}
contains only the genuine system--environment correlations. By construction,
$T_{jk}=0$ for every factorized initial state, whereas correlated initial
states generally satisfy $T_{jk}\neq0$.\\

Since both the unitary evolution
\[
\rho_{SE}(t)=U(t)\rho_{SE}(0)U^\dagger(t)
\]
and the partial trace are linear operations, the reduced Bloch vector depends
linearly on the correlation tensor $T_{jk}$.\\
Evolving $\rho_{SE}(0)$ under $U(t)=e^{-iH_{\mathrm{ex}}t}$ and tracing out $E$ yields
\begin{equation}
\vec{s}(t)=\vec{s}_{\mathrm{pswap}}(t)+\vec{\delta}(t),
\end{equation}
where $\vec{s}_{\mathrm{pswap}}(t)$ is the partial-swap expression in
Eq.~\eqref{eq:pswap} (coming from the $s_j e_k$ part), and
\begin{equation}
\delta_j(t)=\sum_{k,\ell} K_{jk\ell}(t)\,T_{k\ell},
\label{eq:delta}
\end{equation}
for some real coefficients $K_{jk\ell}(t)$ determined by $H_{\mathrm{ex}}$.\\
Equation~\eqref{eq:delta} makes the physical origin of the deviation
explicit: the correction $\vec{\delta}(t)$ is entirely generated by the
correlation tensor. When $T_{jk}=0$, corresponding to an initially
factorized state, the deviation vanishes identically and the reduced
dynamics reduce exactly to the partial-SWAP expression
(\ref{eq:pswap}). Conversely, whenever $T_{jk}\neq0$, the reduced trajectory
generally departs from the partial-SWAP family and cannot, in general, be
reproduced simply by varying the unknown environment Bloch vector.\\
Our witness is precisely sensitive to such deviations: if for some time
$t^{*}$ the measured
$z(t^{*})=s_z(t^{*})$
lies outside the factorized envelope derived from
Eq.~\eqref{eq:pswap}, then necessarily $T\neq0$, and the initial state
$\rho_{SE}(0)$ cannot be written as a product.\\
Thus, the factorized envelope derived in the main text may be viewed as the
set of all trajectories generated under the constraint $T_{jk}=0$. Any
experimentally observed violation of this envelope demonstrates that a
nonzero correlation tensor is required to explain the measured dynamics,
thereby providing a direct certification of initial
system--environment correlations using only the single observable
$z(t)$.

\bibliographystyle{apsrev4-2}
\bibliography{main}

\end{document}